\documentclass[11pt]{article}
\usepackage{fullpage}
\usepackage[english]{babel}
\usepackage{amsfonts}
\usepackage{amsthm}
\usepackage{amssymb}
\usepackage{hyperref}
\usepackage{amsmath}
\usepackage{cite}

\usepackage{graphicx}
\usepackage{enumerate}
\usepackage{enumitem} 
\usepackage{tcolorbox}
\usepackage{soul}
	
\newtheorem{theorem}{Theorem}
\newtheorem{lemma}{Lemma}

\newtheorem{corollary}{Corollary}

\newtheorem{remark}{Remark}

\usepackage{xcolor}

\newenvironment{theorem-repeat}[1]{\begin{trivlist}
		\item[\hspace{\labelsep}{\bf\noindent Theorem \ref{#1}.}]\em }%
	{\end{trivlist}}
	\newenvironment{corollary-repeat}[1]{\begin{trivlist}
		\item[\hspace{\labelsep}{\bf\noindent Corollary \ref{#1}.}]\em }%
	{\end{trivlist}}

\newcommand{\qedsymb}{\qed}

\newcommand{\true}{\texttt{true}}
\newcommand{\false}{\texttt{false}}
\newcommand{\Rvi}{R^{v,2}_{i}}

\newcommand{\IsEmpty}{\texttt{IsEmpty}}
\newcommand{\AreNeighborsEmpty}{\texttt{AreNeighborsEmpty}}
\newcommand{\inconsistent}{\texttt{inconsistent}}

\usepackage{xspace}
\newcommand{\congest}{\mathsf{CONGEST}\xspace}

\usepackage{cancel}

\title{Finding Subgraphs in Highly Dynamic Networks} 

\usepackage{cite}
\begin{document}

\author{Keren Censor-Hillel\thanks{Department of Computer Science, Technion, Haifa, Israel, ckeren@cs.technion.ac.il. This project has received funding from the European Union’s Horizon 2020 Research And Innovation Program under grant agreement no. 755839.} \and Victor I. Kolobov\thanks{Department of Computer Science, Technion, Haifa, Israel, tkolobov@cs.technion.ac.il.} \and Gregory Schwartzman\thanks{Japan Advanced Institute of Science and Technology, greg@jaist.ac.jp. This work was supported by JSPS Kakenhi Grant Number JP19K20216.}}
\date{}
\maketitle

\begin{abstract}
In this paper we consider the fundamental problem of finding subgraphs in highly dynamic distributed networks -- networks which allow an arbitrary number of links to be inserted / deleted per round. We show that the problems of $k$-clique membership listing (for any $k\geq 3$), 4-cycle listing and 5-cycle listing can be deterministically solved in $O(1)$-amortized round complexity, even with limited logarithmic-sized messages.

To achieve $k$-clique membership listing we introduce a very useful combinatorial structure which we name the \emph{robust $2$-hop neighborhood}. This is a subset of the 2-hop neighborhood of a node, and we prove that it can be maintained in highly dynamic networks in $O(1)$-amortized rounds. We also show that maintaining the actual 2-hop neighborhood of a node requires near linear amortized time, showing the necessity of our definition. For $4$-cycle and $5$-cycle listing, we need edges within hop distance 3, for which we similarly define the \emph{robust $3$-hop neighborhood} and prove it can be maintained in highly dynamic networks in $O(1)$-amortized rounds.

We complement the above with several impossibility results. We show that membership listing of any other graph on $k\geq 3$ nodes except $k$-clique requires an almost linear number of amortized communication rounds. We also show that $k$-cycle listing for $k\geq 6$ requires $\Omega(\sqrt{n} / \log n)$ amortized rounds. This, combined with our upper bounds, paints a detailed picture of the complexity landscape for ultra fast graph finding algorithms in this highly dynamic environment. 

\end{abstract}

\section{Introduction}
\label{section:introduction}
Large-scale distributed systems are at the heart of many modern technologies, prime examples being the Internet, peer-to-peer networks, wireless systems, and more. Such environments are inherently subject to dynamic behavior, which in some cases is highly unpredictable. 
For example, highly unpredictable real-world large-scale peer-to-peer networks (up to millions of peers) were studied for a broad selection of applications, such as file-sharing, conferencing, or content distribution~\cite{GummadiSG02,SenW04,StutzbachR06,FalknerPJKA07      ,GummadiDSGLZ03,IsdalPKA10}. It was observed that such networks exhibit a wide range of peer session lengths, ranging from minutes to days, with sessions being short on average but having a \emph{heavy tailed} distribution, demonstrating the heterogenous nature of dynamic peer behavior in the network. Due to the increasing relevance of distributed systems with potentially unpredictable dynamic behavior, there has been abundant research about computing in dynamic distributed networks. 

In this work, we focus on a very harsh setting in which no bound is given on the number or location of links that appear or disappear from the network at a given time, and no structure at all is imposed on the network graph. Such a highly dynamic setting was first studied by Bamberger et al.~\cite{BambergerKM18}, who showed fast algorithms for packing and covering problems, and was then studied by Censor-Hillel et al.~\cite{CDKPS}, who showed fast algorithms for some locally checkable labelings (LCLs). Here, we also adhere to the bandwidth restriction of the latter, allowing only $O(\log n)$ bits to be sent on a link per round.

We focus on subgraph detection problems in highly dynamic distributed networks. Detecting small subgraphs is a fundamental problem in computing, and a particular interest arises in distributed systems. A motivating example is that some tasks admit efficient distributed algorithms in triangle-free graphs~\cite{Hirvonen2017,Pettie2015}. 

\subsection{The network model}
\label{section:preliminaries}
We assume a synchronous network that starts as an empty graph on $n$ nodes and evolves into the graph $G_i = (V_i,E_i)$ at the beginning of round $i$. In each round, the nodes receive indications about the topology changes of which they are part of, for both insertions and deletions. Then, each node can send a message of $O(\log n)$ bits to each of its neighbors.

As a generalization of centralized dynamic data structures, we view a distributed dynamic algorithm as one that maintains a data structure that is distributed among the nodes of the network graph. Any node can be queried at any time for some information. A distributed dynamic data structure needs to respond to a query according to the information it has, without any further communication, in order to avoid stale responses. A crucial difference between the distributed and dynamic settings is that we allow the distributed data structure to also respond that it is in an inconsistent state, in case the neighborhood of the queried node is undergoing topology changes: note that in a centralized setting we can process  topology changes one at a time in order to be able to determine the response to queries, but in the distributed case the topology change may affect the ability to communicate on top of affecting the response itself. Thus, we allow the data structure at a node to indicate that its updating process is still in progress, by responding that it is in an inconsistent state. 

Formally, a distributed dynamic data structure ($DS$) is a data structure that is split among the nodes, such that each node $v$ holds a part $DS_v$. Given a problem $P$, the data structure $DS_v$ at a node $v$ could be queried for a solution for $P$. Upon a query, the data structure $DS_v$ at node $v$ needs to respond with a correct answer \emph{without communication} or indicate that it is in an inconsistent state. Indeed, we refer to Figure~\ref{fig:round} for an illustration of the different stages of our algorithms in the fully dynamic setting.

\begin{figure}[h]
	\begin{center}
		\includegraphics[clip, scale=0.8]{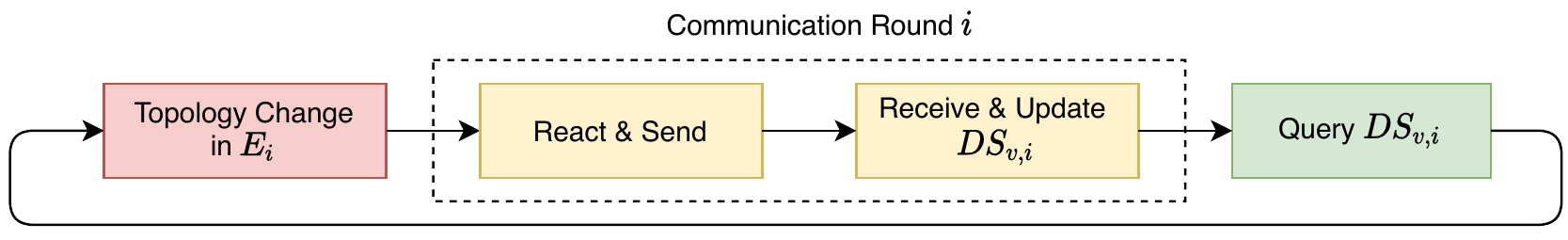}
	\end{center}
	\caption{\label{fig:round} An illustration of the different stages of our algorithms in the fully dynamic setting. At the beginning of a round $i$, there are some topological changes to $E_i$, of which the nodes are \textbf{locally} notified. The communication round that follows is divided into two halves. First, every node $v$ reacts to the changes in $E_i$ by manipulating its local data structure $DS_{v,i}$ in round $i$, followed by sending messages to its neighbors. Next, every node $v$ reads the messages received from its neighbors and updates $DS_{v,i}$ accordingly. At the end of the round, one can query $DS_{v,i}$. The response for the query, which is promised to satisfy some correctness guarantees, must be given immediately, without any further communication.}
\end{figure}

Naturally, an algorithm can trick this definition by always responding that the data structure is in an inconsistent state. However, our complexity measure charges the algorithms when a queried data structure responds that it is inconsistent. Thus, the aim is to design data structures that produce correct responses whenever possible. In general, one may address the \emph{worst-case round complexity} of a distributed dynamic algorithm, which is the maximum number of rounds between the time that a data structure $DS_v$ at a node $v$ becomes inconsistent due to topology changes and the minimum between the time that it becomes consistent again or some additional topology change touches its neighborhood. Crucially, we stress that the problems we consider in this paper \emph{do not admit} algorithms with good worst-case round complexity. Indeed, consider the counterexample where an adversary that starts from an empty graph, connects edges in the graph arbitrarily in the first round and then does no more changes. Since we get an arbitrary graph, an efficient algorithm here for, say, membership listing, would imply one for the $\congest$ model, contradicting the near-linear lower bound $\Omega(n/\log n)$ of~\cite{Izumi2017}.

Nevertheless, we attempt to capture the ``next best thing''. To this end, we consider an \emph{amortized} notion of round complexity, in a similar fashion to the usual notion of amortized complexity, where one divides the total runtime by the number of times a data structure was modified. Since the networks we consider are long-lived environments regardless of whether topology changes occur, we do not count rounds per se, but rather we say that the amortized round complexity of an algorithm is $k$ if for every $i$, until round $i$, the number of rounds in which there exists at least one node $v$ with an inconsistent $DS_v$ divided by the number of topology changes which occurred, is bounded by $k$.  For simplicity, we count the number of topological changes occurring \emph{globally}, although our results hold even if we count the \emph{maximal number of changes occurring at a node}. 

We note that the algorithms of~\cite{CDKPS} can be viewed in the same manner, although for the problems considered there, one cannot rule out algorithms with a good worst-case complexity. We also emphasize that obtaining fast amortized complexity in various dynamic settings is extensively investigated (see, e.g.,~\cite{GuptaK0P17,AbboudA0PS19,BhattacharyaK19,Wajc20,ItalianoLMP19}).

\subsection{Technical contributions}
Our main result is that each node can maintain a list of all triangles to which it belongs, and this can be done in an $O(1)$ amortized number of rounds. 

\newcommand{\TheoremTriangles}
{
There is a deterministic distributed dynamic data structure for triangle membership listing, which handles edge insertions and deletions in $O(1)$ amortized rounds.
}

\begin{theorem}
\label{theorem:triangles}
\TheoremTriangles
\end{theorem}

We note that in~\cite{BonneC19}, the bandwidth required for dynamic distributed algorithms that complete in a single round was investigated. It is shown there that the bandwidth for triangle membership listing is $\Theta(1)$ and $\Theta(n^{1/2})$ for edge deletions and insertions, respectively. The model there is very different from ours: it assumes only one type of topology change, and only one change per round, but does not allow an inconsistent state of the data structure. 

Since membership listing is a very strong guarantee, as opposed to detection or listing, we immediately conclude that a constant amortized round complexity also applies to membership listing of any sized clique. In fact, each node knows all cliques to which it belongs.

\newcommand{\CorollaryCliques}
{
There is a deterministic distributed dynamic data structure for $k$-clique membership listing, for any integer $k \geq 3$, which handles edge insertions and deletions in $O(1)$ amortized rounds.
}

\begin{corollary}
\label{corollary:cliques}
\CorollaryCliques
\end{corollary}

Our techniques do not apply for any subgraph other than cliques, and we show that this is for a good reason. Membership listing of any other subgraph requires a linear number of rounds, even amortized.
\newcommand{\TheoremLBListing}
{
Let $k\geq 3$ be an integer and let $H$ be a $k$-vertex graph which is not the $k$-clique. Then, any deterministic distributed dynamic data structure for $H$ membership listing that handles edge insertions and deletions requires $\Omega\left(\frac{n}{\log n}\right)$ amortized rounds.
}

\begin{theorem}\label{thm:lbmembershiplisting}
\TheoremLBListing
\end{theorem}

Yet, once we relax the membership requirement, we can again find additional subgraphs extremely fast. For the general non-membership listing variant of 4-cycles and 5-cycles we obtain an $O(1)$ amortized complexity.
\newcommand{\TheoremCycles}
{
There is a deterministic distributed dynamic data structure for $4$-cycle listing and $5$-cycle listing, which handles edge insertions and deletions in $O(1)$ amortized rounds.
}

\begin{theorem}
\label{theorem:cycles}
\TheoremCycles
\end{theorem}

Finally, we show that these ultra-fast algorithms are barred at $5$-cycles, namely, that listing of $k$-cycles for larger values of $k$ hits a significant lower bound.

\newcommand{\TheoremCycleLB}
{	Any deterministic distributed dynamic data structure for $k$-cycle listing, for any integer $k\geq 6$, handles edge insertions and deletions in $\Omega\left(\frac{\sqrt{n}}{\log n}\right)$ amortized rounds.
}
\begin{theorem}
\label{theorem:cycleLB}
\TheoremCycleLB
\end{theorem}

\subsection{The challenges and our techniques}
To explain the main challenge in triangle membership listing, as a warm-up, consider a triangle $\{v,u,w\}$. This triangle lies within the 2-hop neighborhood of all of its nodes, where we define the $k$-hop neighborhood of a node $v$ to be the set of edges whose endpoints are within distance $k$ from $v$. A naive approach would thus be for all nodes to learn their 2-hop neighborhood in order to list all triangles which they belong to. However, we prove in Section~\ref{section:cliques} that learning the 2-hop neighborhood of a node is too expensive, requiring a near-linear number of rounds, even in an amortized complexity measure. Thus, the naive algorithm for this problem is prohibitively slow. 

~\\\textbf{Warm-up: robust 2-hop neighborhoods:} 
Despite not being able to learn the entire 2-hop neighborhood, we identify a subset of the 2-hop neighborhood, which we term \emph{the robust 2-hop neighborhood}, which we prove can be maintained by the data structure at each node within $O(1)$ amortized rounds (we give this as a warm-up in Appendix~\ref{section:robust}). This subset of edges consists of all edges adjacent to $v$, and every remaining edge $\{u,w\}$ in the $2$-hop neighborhood of $v$ that is inserted \emph{after} at least one of the edges $\{v,u\}$ and $\{v,w\}$. At a first glance, this task may seem easy: with every insertion of an edge $e=\{v,u\}$, each of its endpoints $v$ enqueues $e$ and sends it to every neighbor $w$ when dequeued (the queue is needed to adhere to the bandwidth restrictions).

However, this is insufficient because the graph may also undergo edge deletions. To illustrate a bad case, let $\{v,u,w\}$ be a triangle in the graph which $v$ has knowledge of. Now, suppose the edge $\{u,w\}$ is deleted from the graph. Let $i_u$ and $i_w$ be the rounds in which $u$ and $w$ tell their neighbors about the deletion of $\{u,w\}$, respectively, and notice that these can be delayed due to congestion caused by previous information that $u$ and $w$ need to send. In particular, assume, crucially, that $i_w\neq i_u$. Now, suppose the edge $\{v,u\}$ is deleted in round $i_u$, and the edge $\{v,w\}$ is deleted in round $i_w$, and both are immediately inserted back in the following rounds. This causes $v$ to never learn about the deletion of $\{u,w\}$ but $\{u,w\}$ will still be marked as existing according to $v$, since at least one of the edges $\{v,u\}$, $\{v,w\}$ was present in every round. Thus this algorithm fails.

Therefore, to fix this discrepancy, we need to make $v$ remove the edge $\{u,w\}$ from the set of edges that it knows about, also in the first case. We overcome this issue by assigning timestamps to edge insertions, and requiring that $v$ \emph{forgets} about the edge $\{u,w\}$ if its insertion time becomes smaller than both insertion times of $\{v,u\}$ and $\{v,w\}$. This solves the bad scenario above, as $v$ will forget the edge $\{u,w\}$ once $\{v,w\}$ is deleted. A major problem here is that insertion times grow unboundedly, preventing us from sending them in messages. We overcome this by having a node $v$ that receives an edge $\{u,w\}$ through its neighbor $u$ assign an imaginary timestamp $\{u,w\}$ that is equal to the insertion time of $\{v,u\}$, rather than that of $\{u,w\}$ itself, as the latter is unknown to $v$, and we prove that these imaginary timestamps are sufficient.

~\\\textbf{Triangle membership listing:} 
Having the robust 2-hop neighborhood is very useful for listing triangles, but care should be taken. Suppose the edges of the  triangle $\{v,u,w\}$ are inserted in different rounds, in the order $\{v,u\}$, $\{u,w\}$, $\{v,w\}$. Then nodes $v$ and $u$ know about this triangle because it is contained in their robust 2-hop neighborhoods. Yet, the node $w$ does not know about the edge $\{v,u\}$. 

To address this, we would like to let $v$ inform $w$ about the triangle created by insertion of the edge $\{v,w\}$. This raises a subtlety: if $v$ is responsible for informing $w$ about $\{v,u,w\}$, it is also responsible for informing $w$ about any other triangle $\{v,u',w\}$ that the insertion of the edge $\{v,w\}$ creates, but there could be a linear number of such triangles which would cause a huge congestion on the communication link $\{v,w\}$. Instead, we let the node $u$ inform $w$ of this triangle, noticing that each such $u$ only sends one indication of a triangle to $w$ per the inserted edge $\{v,w\}$. Note that we incur here another round towards the amortized round complexity, because $u$ only knows about the edge $\{v,w\}$ after at least another round since it needs to obtain this information from $v$ or $w$. But for the amortization argument this will turn out to be sufficient (see Section~\ref{section:cliques}).

~\\\textbf{Robust 3-hop neighborhoods and listing 4-cycles and 5-cycles:} 
Some 4-cycles can be listed using the robust 2-hop neighborhood, e.g., if $v$ is a node on the cycle $v-u-w-x$ whose edges $\{v,u\}$ and $\{v,x\}$ are inserted before the other two edges. However, once we consider an order of insertions that is $\{v,u\}$, $\{w,x\}$, $\{v,x\}$, and $\{u,w\}$, the 4-cycle is not contained in the robust 2-hop neighborhood of any of its nodes. For 5-cycles, the above is always the case, regardless of the order of insertions.

Thus, in order to list 4-cycles and 5-cycles, we need to know more edges. For this, we identify a subset of the 3-hop neighborhood of $v$, which we term \emph{the robust 3-hop neighborhood}, for which we can prove that (i) $v$ can maintain knowledge of this subset within $O(1)$ amortized rounds (Section~\ref{section:2hop}), and (ii) this subset is sufficient for constructing a data structure for $4$-cycle and $5$-cycle listing (Section~\ref{section:cycles}). We note that this includes the robust 2-hop neighborhood. One could also consider defining this for larger hops, but note that listing of larger cycles admits a significant lower bound, as we prove in Section \ref{section:cycleLB}.

Roughly speaking, this subset of the 3-hop neighborhood of a node $v$ will consist of edges on 2,3-paths towards $v$ such that the farthest edge is inserted after the other(s). It will not be hard to show that with insertions only, the knowledge of this subset of edges can be efficiently transmitted to $v$. When we allow edge deletions, we require the deletion to be propagated to distance $2$. 

While this approach sounds straightforward, several problems still remain. First, it might make the knowledge that $v$ maintains be a disconnected graph, which is a pitfall we wish to avoid, since $v$ can never faithfully maintain information about unreachable components in the graph. Consequently, one could require that $v$ simply forgets about now unreachable components. However, we argue that this is not sufficient.
Indeed, this problem appears already for distance $1$ as described in the $2$-hop case, where in a triangle, the far edge from $v$ is deleted, but this information fails to reach $v$ due to flickering of the two edges of the triangle that touch $v$. When limiting ourselves to the $2$-hop case, we overcame this by using (imaginary) timestamps for edge insertions. For the $3$-hop case this is insufficient, and instead we employ a more involved mechanism of maintaining a set of paths.

Formally, we have $v$ maintain for each edge $\{w,x\}$ \emph{a set of paths on which the edge was learned}. If $\{w,x\}$ is rediscovered on a new path, this path is added into the set, and if, alternatively, a path is severed due to an edge deletion, this path is removed from the set. If no path remains for the edge $\{w,x\}$, only then is it marked as not existing. We claim that this algorithm is sufficient for a node $v$ to learn the required subset of its 3-hop neighborhood even when allowing edge insertions and deletions. The proof is highly non-trivial, especially due to the need to argue about the amortized round complexity. To get a flavor of this, note that a similar approach of using the \emph{robust $4$-hop neighborhood listing} in order to obtain $6$-cycle listing is doomed to fail, given our lower bound of Section~\ref{section:cycleLB}.

\subsection{Additional related work}
There are two previous works that address the recently emerging highly dynamic setting that we address. Pioneering this area was~\cite{BambergerKM18}, which studied the complexity of packing and covering problems in this setting, and this was followed by the work in~\cite{CDKPS}, which addresses maximal matching, coloring, maximal independent sets, and 2-approximations for weighted vertex cover. Our results for subgraphs are the first in this model.

Many additional models of dynamic distributed computing have been extensively studied throughout the years. A prime example is the literature about self stabilization~\cite{Dolev2000}, which also addresses a notion of \emph{quiet} rounds. The highly dynamic setting considered in our work does not rely on any quiet time in the network. A very harsh model that allows the graph to almost completely change from round to round is that of~\cite{Kuhn2010} (see also follow-up work), but then the questions addressed have a flavor of information dissemination rather than graph properties. Additional work that addresses graph structures in distributed dynamic settings include the aforementioned clique detection work~\cite{BonneC19}, whose setting differs from ours by not allowing inconsistency responses but sometimes reverting to an increased bandwidth. The latter also assume only a single topology changes per round. Studies by~\cite{KonigW13, Censor-HillelHK16, AssadiOSS18, Solomon16, ParterPS16} assume enough time for the network to produce a response after a topology change, and some of these works also allow a large bandwidth. In~\cite{AugustinePRU12,AugustineP0RU15,AugustineMMPRU13}, a highly dynamic model for peer-to-peer networks is studied, in which the graphs at each time must be bounded-degree expanders, but heavy churn (rate of peers joining and leaving) is allowed. 

To contrast our results about subgraphs with the static $\congest$ model, note that the round complexity of triangle membership listing is $\tilde{\Theta}(n)$ \cite{Izumi2017}. The complexity of triangle listing (every triangle needs to be known by some node, but not necessarily by all of its nodes) is $\tilde\Theta(n^{1/3})$, due to the upper bound of~\cite{ChangS19} (obtained after a series of papers that were able to show that it is sublinear~\cite{Izumi2017,Chang2019}) and the lower bounds of~\cite{Izumi2017, Pandurangan0S18}. If we consider graphs with bounded maximum degree $\Delta$, then the complexity of triangle listing is $O(\Delta/\log n+\log\log\Delta)$ \cite{HuangPZZ20}. In particular, this is superior to the previous complexity of $\tilde\Theta(n^{1/3})$ whenever $\Delta=\tilde O(n^{1/3})$. Furthermore, restricted to \emph{deterministic} algorithms, the complexity was recently shown to be $n^{2/3+o(1)}$ \cite{ChangS20}. The complexity of triangle detection (some node needs to indicate that there is a triangle) is $\tilde{O}(n^{1/3})$ as listing, and the lower bound front is very scarce: it is known that a single round is insufficient~\cite{Abboud2017} even for randomized algorithms~\cite{Fischer2018}. For deterministic algorithms, the complexity was recently shown to be $n^{1-1/\omega+o(1)}<O(n^{0.58})$ \cite{ChangS20}, where $\omega$ is the matrix multiplication exponent.
In the $\congest$ model, listing cliques of size $k$ admits a lower bound of $\tilde{\Omega}(n^{k-2/k})$ \cite{Fischer2018} (for $k=4$ this also follows from~\cite{Czumaj2018}), and upper bounds of $O(n^{5/6+o(1)})$ and $O(n^{21/22+o(1)})$ are known for listing of 4-cliques and 5-cliques, respectively~\cite{EdenFFKO19}. Recently, sublinear running times were reported for listing all $k$-cliques~\cite{CHLL20}. 
Membership listing has a linear complexity as is the case for triangles.
Listing $4$-cycles requires a linear number of rounds~\cite{EdenFFKO19}, but the complexity of detecting $4$-cycles is $\Theta(n^{1/2})$, given in~\cite{Drucker2014}, and the complexities of detecting larger cycles are also sublinear but known to be polynomial in $n$ \cite{EdenFFKO19}.

\section{Clique membership listing in $O(1)$ amortized complexity}
\label{section:cliques}

The core of showing clique membership listing is showing membership listing for triangles. The problem of \emph{triangle membership listing} requires the data structure $DS_v$ at each node $v$ to respond to a query of the form $\{v,u,w\}$ with an answer $\true$ if this set forms a triangle, $\false$ if it does not, or $\inconsistent$, if $DS_v$ is in an inconsistent state. Recall that the node $v$ is not allowed to use any communication for deciding on its response.

Note that the above would be a trivial task if large messages were available, by simply having each node send its entire neighborhood after each topology change.

\begin{theorem-repeat}{theorem:triangles}
\TheoremTriangles
\end{theorem-repeat}

Let $E^{v,r}_i$ denote the subset of $E_i$ of all edges contained in the $r$-hop neighborhood of $v$. Consider first only the case $r=2$, for which $E^{v,2}_i$ is the set of edges that touch the node $v$ or any of its neighbors. One would like to have each node $v$ learn all of $E^{v,2}_i$ but this turns out to be a hard task, as we show later in Corollary \ref{cor:lb1hop}\footnote{Furthermore, in Appendix \ref{appendix:1hop} we show an optimal algorithm for 2-hop neighborhood listing.}. Instead, in Theorem \ref{theorem:N1} in Appendix~\ref{section:robust}, which we refer to for a warm-up, it is shown that a node $v$ may learn a subset of $E_{i}^{v,2}$ (the \emph{robust 2-hop neighborhood}), subject to insertion time constraints. We wish to extend this notion and ask which other subsets of $E_{i}^{v,2}$ (or more generally $E_{i}^{v,r}$)  that are subject to insertion time constraints a node $v$ can maintain knowledge of. We will refer to these sets of edges with insertion time constraints as \emph{temporal edge patterns}. In Figure \ref{fig:teptriangle} we present temporal edge patterns the knowledge of which a node $v$ maintains in the proof of Theorem \ref{theorem:triangles}. These patterns, together with $N_v$, are sufficient for triangle membership listing.

\begin{figure}[h]
	\begin{center}
		\includegraphics[clip, scale=1]{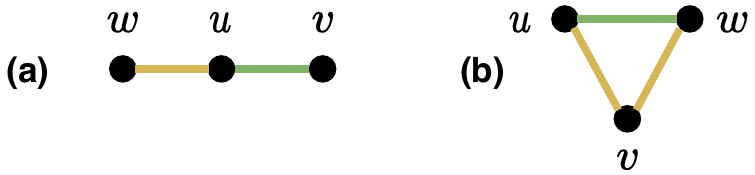}
	\end{center}
	\caption{\label{fig:teptriangle}The temporal edge patterns a node $v$ learns in the proof of Theorem \ref{theorem:triangles}. These patterns consist of \textbf{(a)} the robust $2$-hop neighborhood where $t_{\{u,w\}}\geq t_{\{v,u\}}$, and  \textbf{(b)} a pattern where $t_{\{u,w\}}< t_{\{v,u\}},t_{\{v,w\}}$. Note that patterns \textbf{(a)} and \textbf{(b)} are mutually disjoint. Green colored edges have earlier insertion time compared to yellow colored edges.}
\end{figure}

\begin{proof}[Proof of Theorem \ref{theorem:triangles}] 
	
Denote by $T_{i}^{v,2}\subseteq E_{i}^{v,2}$ the set of edges which satisfy the time constraints given in Figure \ref{fig:teptriangle}, namely, if the edges $\{v,u\},\{u,w\}$ exist then we have $\{v,u\},\{u,w\}\in T_{i}^{v,2}$ if and only if either \textbf{(a)} $t_{\{u,w\}}\geq t_{\{v,u\}}$, or \textbf{(b)} the edge $\{v,w\}$ exists and $t_{\{u,w\}}< t_{\{v,u\}},t_{\{v,w\}}$. Furthermore, we include all the edges incident at $v$ into $T_i^{v,2}$, that is, $E^{v,1}_i\subseteq T_i^{v,2}$. The data structure $DS_{v,i}$ at node $v$ at the end of round $i$ consists of the following: A set $S_{v,i}$ of items, where each item is an edge $e=\{a,b\}$ along with a timestamp $t'_e$, a queue $Q_{v,i}$ of items, each of which is either 

\begin{enumerate}
	\item an edge $e=\{u,w\}$ along with an insertion/deletion mark and a mark that this item is related to temporal edge pattern \textbf{(a)};
	\item a tuple of an edge and a vertex $\langle\{u,w\},x\rangle$ along with an insertion/deletion mark and a mark that this item is related to temporal edge pattern \textbf{(b)},
\end{enumerate}
and a flag $C_{v,i}$. We make a distinction between $t_e$ which is the true timestamp when $e$ was added and $t'_e$ which is the imaginary timestamp maintained inside $S_{v,i}$. Note that for every $e$ adjacent to $v$, the node $v$ knows the value $t_e$, while for non adjacent edges we only know $t'_e$.
	We will make sure that an edge $\{a,b\}$ appears in at most a single item in $S_{v,i}$. Our goal is to maintain that $T^{v,2}_i= S_{v,i} $ at the end of round $i$, or that $C_{v,i}=\false$ (the consistency flag). 
	
	Initially, for all nodes $v$, we have that $S_{v,0}$ and $Q_{v,0}$ are empty, and $C_{v,0}=\true$, indicating that $DS_{v,0}$ is consistent.
	The algorithm for a node $v$ in round $i\geq 1$ is as follows:
	\begin{enumerate}
		\item \label{edge:init2} Initialization: Set $S_{v,i} = S_{v,i-1}$ and $Q_{v,i} = Q_{v,i-1}$. 
		
		\item \label{edge:chng2} Topology changes: Upon indications of edge deletions, for each such deletion $\{v,u\}$, the edge $\{v,u\}$ is removed from $S_{v,i}$. 
		Then, for each such deletion $\{v,u\}$, all edges $\{\{u,z\}\in S_{v,i} \mid \text{either }\{v,z\} \not\in S_{v,i} \text{ or }t'_{\{u,z\}} < t_{\{v,z\}} \}$ are removed from $S_{v,i}$. 
		Afterwards, upon an indication of an edge insertion $\{v,u\}$, the edge $\{v,u\}$ is added to the set $S_{v,i}$.
		In both cases (insertion and deletion), the pair $\{v,u\}$ is enqueued into $Q_{v,i}$ along with a corresponding insertion/deletion mark and mark \textbf{(a)}. 
		
		\item \label{edge:comm2} Communication: 
		If $Q_{v,i}$ is not empty, the node $v$ dequeues an item $e$ from $Q_{v,i}$. If the item has mark \textbf{(a)}, $v$ sends it with mark \textbf{(a)} to all of its neighbors $u$ such that $t_e \geq t_{\{ v,u\}}$. Alternatively, if the item has mark \textbf{(b)}, it has the form $\langle\{v,u\},w\rangle$, in which case we send $\{v,u\}$ to $w$ with mark \textbf{(b)}.
		In any case $v$ also sends to all its neighbors a Boolean indication $\IsEmpty$ of whether $Q_{v,i}$ was empty in the beginning of round $i$. In actuality, we do not send $\IsEmpty=\true$: not receiving $\IsEmpty=\false$ by other nodes is interpreted as receiving $\IsEmpty=\true$.
		
		\item \label{edge:updt2} Updating the data structure: Upon receiving an item $e=\{u,w\}$ with mark \textbf{(a)} from a neighbor $u$, node $v$ sets updates $S_{v,i}$ according to the insertion/deletion mark. Furthermore, for the case of insertion, if $e\notin S_{v,i-1}$ we set $t'_e=t_{\{u,v\}}$. Otherwise we update $t'_e=\max \{t'_e, t_{\{u,v\}}\}$. If we further detect that  $t_{\{v,u\}}<t_{\{v,w\}}\leq t'_e$ or $t_{\{v,w\}}< t_{\{v,u\}}\leq t'_e$, we enqueue $\langle\{v,u\},w\rangle$ or $\langle\{v,w\},u\rangle$, respectively, into $Q_{v,i}$ with mark \textbf{(b)}. 
		
		Alternatively, upon receiving an item $e=\{u,w\}$ with mark \textbf{(b)} from a neighbor $u$, only if the edges $\{v,u\},\{v,w\}$ exist, $v$ inserts it into $S_{v,i}$ and sets $t'_e=\min\{t_{\{v,u\}},t_{\{v,w\}}\}-1$.
		
	If $Q_{v,i}$ is not empty, or an item with $\IsEmpty=\false$ is received, then $C_{v,i}$ is set to $\false$. Otherwise, $C_{v,i}$ is set to $\true$.
	\end{enumerate}
	~\\\textbf{Correctness:} Suppose that $C_{v,i} = \true$ and $DS_{v,i}$ is queried with $\{v,u,w\}$. We need to show that $DS_{v,i}$ responds $\true$ if and only if this triplet is a triangle in the graph $G_i$. We will instead show that we implement a data structure such that it responds $\true$ on a query $\{u,w\}$ if and only if $\{u,w\}\in S_{v,i}$ where we promise $T_{i}^{v,2}=S_{v,i}$. This is sufficient, since $\{v,u,w\}$ is a triangle in $G_i$ if and only if $\{v,u\},\{u,w\},\{v,w\}\in T_{i}^{v,2}$. Now, we first prove correctness for the ideal algorithm, where we set $t'_e = t_e$ (the true edge timestamp). Then, we show that our algorithm behaves exactly the same as the ideal algorithm, even with the modified timestamps.
	
	To show the correctness of the ideal algorithm, we want to show that if $C_{v,i}=\true$, it holds that $S_{v,i}=T^{v,2}_i$. Let $e=\{u,z\}\in T^{v,2}_i$, and we show that if $C_{v,i}=\true$ then $e$ is also in $S_{v,i}$. If $v\in e$ the proof is direct, so we focus on the case where $v \notin e$. Because $e\in T^{v,2}_i$ either \textbf{(a)} there must exist an edge $e'=\{u,v\} \in T^{v,2}_i$ such that $t_{e} \geq t_{e'}$ or \textbf{(b)} both edges $\{u,v\},\{z,v\}$ exist and $t_e<t_{\{u,v\}},t_{\{z,v\}}$. In case \textbf{(a)}, the edge $e'$ was added at the same round or before $e$, and remained there until the $i$-th round. As the $C_{v,i}=\true$ implies that the queue of $u$ is empty, this means that $v$ must have received the edge $e'$ by iteration $i$ and added it to $S_{v,i}$. As $e'$ was unchanged throughout this time, this implies that the edge $e$ was not deleted from $S_{v,i}$, up to and including round $i$. In case \textbf{(b)}, the edge $e$ was inserted strictly before $\{v,u\}$ and $\{v,z\}$, and all three exist until round $i$. Suppose, without loss of generality, that $t_{\{v,z\}}\geq t_{\{v,u\}}>t_e$. In this case, by step \ref{edge:updt2}, after learning about $\{v,z\}$, $u$ should enqueue an element $\langle e,v\rangle$ which then by step \ref{edge:comm2} implies that $v$ should eventually learn of $e$. We first show that $u$ learns about $\{v,z\}$. Indeed, both $v$ and $z$ enqueue $\{v,z\}$ and as $C_{v,i}=\true$ implies these queues are empty we conclude $u$ learns about that edge. Note also that $u$ learns about $\{v,z\}$ when $\{u,v\}$ and $\{u,z\}$ are already present. Next, as all the queues around $v$ are eventually empty, $u$ detects that $t_{\{v,z\}}\geq t_{\{v,u\}}>t_e$, which causes it to enqueue $\langle e,v\rangle$. Thus $e$ is eventually sent to $v$ by $u$, and so $e\in S_{v,i}$.
	
Next, let $e=\{u,z\}\notin T^{v,2}_i$, and we show that if $C_{v,i}=\true$ then $e$ is also  not in $S_{v,i}$. Here again the case where $v\in e$ is trivial. If $e$ was never added to $S_{v,j}$ at some round $j<i$ we have $e\notin S_{v,i}$ and we are done. Thus, we are only concerned with the case where $e$ was added to $S_{v,j}$ for $j<i$, but is not in $T^{v,2}_i$. This implies that there exists no $e'\in T^{v,2}_i$ such that $e\cap e' \neq \emptyset$ and $t_{e} \geq t_{e'}$. By step~\ref{edge:chng2} of our algorithm, this implies that $e$ must be deleted by the $i$-th iteration. Furthermore, by the same reasoning, it cannot be that both edges $\{v,u\},\{v,z\}$ exist and $t_e<t_{\{v,u\}},t_{\{v,z\}}$. If $e$ was received by $v$ when both edges exist, by step~\ref{edge:chng2} of our algorithm, this implies that $e$ must be deleted by the $i$-th iteration. If, alternatively, $e$ was received by $v$ when one of $\{v,u\},\{v,z\}$ did not exist already, then by step~\ref{edge:updt2} it would not be included into $S_{v,i}$ in the first place. This completes the proof for the ideal algorithm.
	
	We need to show that the set $S_{v,i}$ maintained by our algorithm is the same as it would be if the value of $t'_e$ was set to $t_e$ (the ideal algorithm). We observe that exact timestamps are not required -- we only need to know how they relate to one another to be able to decide if edges belong to a certain temporal edge pattern or not. This is made sure in step~\ref{edge:updt2}. Indeed, notice that in step 1 we have $t'_{\{u,z\}} < t_{\{v,z\}}$ if and only if $t_{\{u,z\}} < t_{\{v,z\}}$, and in step 4 we have $t_{\{v,u\}}<t_{\{v,w\}}\leq t'_e$ if and only if $t_{\{v,u\}}<t_{\{v,w\}}\leq t_e$ and $t_{\{v,w\}}< t_{\{v,u\}}\leq t'_e$ if and only if $t_{\{v,w\}}< t_{\{v,u\}}\leq t_e$. Therefore the algorithm operates in the same manner as the ideal algorithm.
	
	~\\\textbf{Round complexity:}
	A topology change in the pair $e=\{w,u\}$ in round $i$ causes an enqueue of an item to $Q_{u,i}$ and to $Q_{w,i}$ with mark \textbf{(a)}. For nodes $v_1,\ldots,v_\ell$ which are neighbors of both $u$ and $w$ this can further cause an enqueue, at most twice (at rounds $i_u,i_w>i$), of an item with mark \textbf{(b)} to $Q_{v_1},\ldots,Q_{v_\ell}$, respectively. Since $e$ caused an enqueue on at most $3$ rounds and nodes dequeue a single element every round we then have, for every round $j$, that the number of rounds in which there exists at least one node $v$ with an inconsistent $DS_v$ until round $j$ is bounded by $3$ times the number of topology changes which occurred until round $j$. 
	This gives the claimed $O(1)$ amortized round complexity.
\end{proof}

Given an integer $k \geq 3$, the \emph{$k$-clique membership listing} problem is the natural generalization of the triangle membership listing problem, in which the query to $DS_v$ is a $k$-sized set $H=\{v_1,\dots,v_k\}$, where there is an $1\leq i\leq k$ such that $v_i=v$. Now observe that triangle membership listing implies $k$-clique membership listing, because for any $k$-clique $H$ with a node $v$, if $v$ knows about all triangles of which it is a member, then it knows about all edges in $H$. We thus immediately obtain the following.

\begin{corollary-repeat}{corollary:cliques}
\CorollaryCliques
\end{corollary-repeat}

Finally, we show that $k$-cliques are essentially the only subgraphs for which we can handle membership listing efficiently.

\begin{theorem-repeat}{thm:lbmembershiplisting}
\TheoremLBListing
\end{theorem-repeat}

\begin{proof}
	Let $a$ and $b$ be two vertices in $H$ which are not neighbors, and denote by $N_a$ and $N_b$ their neighborhoods in $H$, respectively. Furthermore, denote by $v_1,v_2,\ldots,v_{k-2}$ all the other vertices of $H$. We will consider the counter example on $n$ nodes where we have some nodes $v_1,v_2,\ldots,v_{k-2}$ which are connected according to $H$. The adversary performs the following steps for $\ell=1,\ldots,t$:
	\begin{enumerate}
		\item Choose a node $u_\ell$ arbitrarily which is different from $v_1,v_2,\ldots,v_{k-2}$ and $u_1,u_2,\ldots,u_{\ell-1}$.
		\item Connect $u_\ell$ to $v_1,v_2,\ldots,v_{k-2}$ according to $N_a$.
		\item Wait for the algorithm to stabilize.
		\item Disconnect $u_\ell$ from all nodes and connect it again according to $N_b$.
 
	\end{enumerate}
	We claim that no algorithm can handle this scenario with $o\left(\frac{n}{\log n}\right)$ amortized round complexity. Indeed, consider the $\ell$th node $u_\ell$ that we connect according to $N_a$. Due to an indistinguishability argument, there are at least $\binom{n-k+1}{\ell-1}$ possible $H$-graphs that $u_\ell$ may form with the nodes $v_1,v_2,\ldots,v_{k-2}$ and each possible selection of $u_1,u_2,\ldots,u_{\ell-1}$. Therefore, at least $\log \binom{n-k+1}{\ell-1}$ bits need to be communicated on the existing edges out of $\{u_\ell,v_1\},\ldots,\{u_\ell,v_{k-2}\}$. Suppose we continue this procedure up to $t=1+\frac{n-k+1}{2}$. Then, the total communication turns out to be at least
	\begin{align}
	\sum_{\ell=1}^{1+\frac{n-k+1}{2}}\log \binom{n-k+1}{\ell-1}&\geq\sum_{\ell=1+\frac{n-k+1}{3}}^{1+\frac{n-k+1}{2}}\log \binom{n-k+1}{\ell-1}
	\geq \sum_{\ell=1+\frac{n-k+1}{3}}^{1+\frac{n-k+1}{2}}\log \left(\frac{n-k+1}{\ell-1}\right)^{\ell-1}\nonumber\\
	&\geq\frac{n-k+1}{6}\cdot\log 3^{\frac{n-k+1}{3}}
	=\Omega(n^2).\nonumber
	\end{align}
	While the communication happens on $O(kn)=O(n)$ different edges, it happens sequentially on at most $O(k)=O(1)$ edges at a time, in step 3, each having $O(\log n)$ bandwidth. Thus we argue this scenario requires at least $\Omega\left(\frac{n^2}{\log n}\right)$ inconsistent rounds. As there are only $O(kn)=O(n)$ topological changes, the amortized round complexity is at least $\Omega\left(\frac{n}{\log n}\right)$.
\end{proof}

Note that maintaining knowledge of the 2-hop neighborhood is nothing else than membership listing of the 3-vertex path graph, $v-u-w$. We deduce the following.

\begin{corollary}\label{cor:lb1hop}
	Any deterministic distributed dynamic data structure for $2$-hop neighborhood listing handles edge insertions/deletions in $\Omega\left(\frac{n}{\log n}\right)$ amortized rounds.
\end{corollary}

This matches the upper bound shown in Appendix \ref{appendix:1hop}.

\section{Listing $4$-cycles and $5$-cycles in $O(1)$ amortized complexity}
\label{section:cycles}
For the problem of $4$-cycles and $5$-cycles, we address the listing variant (not membership listing). This means that for every $4$-cycle ($5$-cycle) $H$, we require that at least one node $v$ in $H$ which, if queried for $H$, returns $\true$. Formally, the problem of \emph{$4$-cycle ($5$-cycle) listing} requires the data structure $DS_{v,i}$ at each node $v$ to respond to a query of the form $H=\{v,u_1,\ldots,u_3\}$ ($H=\{v,u_1,\ldots,u_4\}$) with an answer $\true$, $\false$, or $\inconsistent$, such that if all nodes of $H$ are queried, then either at least one node responds $\inconsistent$, or at least one node outputs $\true$ if and only if $H$ is a $4$-cycle ($5$-cycle) in $G_{i-1}$\footnote{We require correctness with respect to $G_{i-1}$ and not $G_i$ because of the inherent delay of topological changes on edges touching nodes within distance $3$.}.

\begin{theorem}\label{theorem:45cycles}
	There are deterministic distributed dynamic data structures for $4$-cycle and $5$-cycle listing, which handle edge insertions/deletions in $O(1)$ amortized rounds.
\end{theorem}	

In similar fashion to the previous section, we wish to characterize a temporal edge pattern $R^{v,3}_{i}\subseteq E^{v,3}_{i}$, which $v$ can maintain, upon which we can construct a data structure for $4$-cycle and $5$-cycle listing. We will choose $R^{v,3}_{i}$ (the \emph{robust $3$-hop neighborhood}) to be the set of edges which satisfy the time constraints given in Figure \ref{fig:tepcycle}.
Namely, we have
\begin{itemize}
	\item \textbf{(a)} $v-u-w\subseteq R^{v,3}_{i}$ if $t_{\{u,w\}}\geq t_{\{v,u\}}$;
	\item \textbf{(b)} $v-u-w-x\subseteq R^{v,3}_{i}$ if $t_{\{w,x\}}\geq t_{\{u,w\}},t_{\{v,u\}}$.
\end{itemize}
Furthermore, we include all the edges incident at $v$ into $R_i^{v,3}$, that is, $E^{v,1}_i\subseteq R_i^{v,3}$.
\begin{figure}[h]
	\begin{center}
		\includegraphics[clip, scale=1]{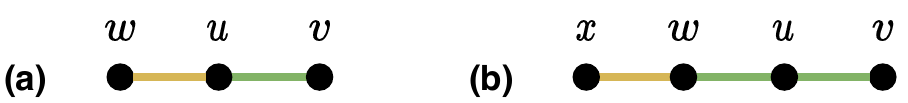}
	\end{center}
	\caption{\label{fig:tepcycle}The temporal edge patterns a node $v$ learns in the proof of Theorems \ref{theorem:modified2robust} and \ref{theorem:45cycles}. These patterns consist of \textbf{(a)} the robust $2$-hop neighborhood where $t_{\{u,w\}}\geq t_{\{v,u\}}$, and  \textbf{(b)} a pattern where $t_{\{w,x\}}\geq t_{\{u,w\}},t_{\{v,u\}}$. Note that patterns \textbf{(a)} and \textbf{(b)} are \textbf{not} mutually disjoint. Green colored edges have earlier insertion time compared to yellow colored edges.} 
\end{figure}

Formally, the \emph{robust $3$-hop neighborhood listing} problem requires the data structure $DS_{v,i}$ at round $i$ at each node $v$ to respond to a query of the form $\{u,w\}$ with an answer $\true$ if $\{u,w\}\in R^{v,3}_{i-1}$, $\false$ if $\{u,w\}\notin E^{v,3}_{i-1}$, or $\inconsistent$, if $DS_v$ is in an inconsistent state. We claim the following.

\begin{theorem}\label{theorem:modified2robust}
	There is a deterministic distributed dynamic data structure for the robust $3$-hop neighborhood listing, which handles edge insertions/deletions in $O(1)$ amortized rounds.
\end{theorem}

Before proving Theorem \ref{theorem:modified2robust}, we show how it implies Theorem \ref{theorem:45cycles}.

\begin{proof}[Proof of Theorem \ref{theorem:45cycles}]
By Theorem \ref{theorem:modified2robust}, suppose for each node $v$ we have a data structure $DS_{v}$ such that at round $i$ it is able to respond to a query of the form $\{u,w\}$ with an answer $\true$ if $\{u,w\}\in R^{v,3}_{i-1}$, $\false$ if $\{u,w\}\notin E^{v,3}_{i-1}$, or $\inconsistent$, if $DS_v$ is in an inconsistent state. In addition, let $k\in\{4,5\}$. We claim it is possible for a node $v$, given a query of the form $H=\{v,u_1,\ldots,u_{k-1}\}$, to respond with an answer $\true$, $\false$, or $\inconsistent$, such that if all nodes of $H$ are queried, then either there is a response that is $\inconsistent$, or at least one node outputs $\true$ if and only if $H$ is a $k$-cycle in $G_{i-1}$. The new data structure will respond $\true$ on $H=\{v,u_1,\ldots,u_{k-1}\}$ if and only if the old data structure responded $\true$ on all $\{v,u_1\},\{u_1,u_2\},\ldots,\{u_{k-2},u_{k-1}\},\{u_{k-1},v\}$. We divide the proof into two cases.

\begin{enumerate}
	\item $k=4$: Suppose $H=\{v,u_1,u_2,u_3\}$ is a $4$-cycle in $G_{i-1}$ and assume, without loss of generality, that $t_{\{u_2,u_3\}}\geq t_{\{v,u_1\}},t_{\{u_1,u_2\}},t_{\{u_3,v\}}$. It is easy to verify that by the definition of $R^{v,3}_{i-1}$ it holds that $\{u_2,u_3\},\{v,u_1\},\{u_1,u_2\},\{u_3,v\}\in R^{v,3}_{i-1}$. Alternatively, suppose that $H=\{v,u_1,u_2,u_3\}$ is not a $4$-cycle in $G_{i-1}$, which implies one of its edges, $e$, is missing. Then the claim follows since a consistent $DS_{v}$ responds $\false$ on edges $e\notin E^{v,3}_{i-1}$.
	\item $k=5$: Suppose $H=\{v,u_1,u_2,u_3,u_4\}$ is a $5$-cycle in $G_{i-1}$ and assume, without loss of generality, that $t_{\{u_2,u_3\}}\geq t_{\{v,u_1\}},t_{\{u_1,u_2\}},t_{\{u_3,u_4\}},t_{\{u_4,v\}}$. It is easy to verify that by the definition of $R^{v,3}_{i-1}$ it holds that $\{u_2,u_3\},\{v,u_1\},\{u_1,u_2\},\{u_3,u_4\},\{u_4,v\}\in R^{v,3}_{i-1}$. Alternatively, suppose that $H=\{v,u_1,u_2,u_3,u_4\}$ is not a $5$-cycle in $G_{i-1}$, which implies one of its edges, $e$, is missing. Then the claim follows since a consistent $DS_{v}$ responds $\false$ on edges $e\notin E^{v,3}_{i-1}$.
\end{enumerate}
\vspace{-0.4cm}
\end{proof}

\section{The robust 3-hop neighborhood}
\label{section:2hop}

We now provide the proof of our algorithm for obtaining the robust 3-hop neighborhood listing. We refer the reader to the introduction for an informal description of this construction.

\begin{proof}[Proof of Theorem \ref{theorem:modified2robust}]
The data structure $DS_{v,i}$ at node $v$ at the end of round $i$ consists of the following: A set $S_{v,i}$ of items, where each item is an edge $e$ along with a set of paths $P_e$, a queue $Q_{v,i}$ of items, each of which is either a path with an insertion mark or an edge with an $O(1)$ bit number and a deletion mark, and a flag $C_{v,i}$. For convenience, denote by $\tilde{S}_{v,i}$ the set of edges in $S_{v,i}$, such that each edge $e$ has nonempty $P_e$. Our goal is to maintain that $R^{v,3}_{i-1}\subseteq \tilde{S}_{v,i}\subseteq E^{v,3}_{i-1}$ at the end of round $i$, or that $C_{v,i}=\false$ (the consistency flag). 

Initially, for all nodes $v$, we have that $S_{v,0}$ and $Q_{v,0}$ are empty, and $C_{v,0}=\true$, indicating that $DS_{v,0}$ is consistent.
The algorithm for a node $v$ in round $i\geq 1$ is as follows:
\begin{enumerate}
	\item \label{edge:init3} Initialization: Set $S_{v,i} = S_{v,i-1}$ and $Q_{v,i} = Q_{v,i-1}$. 
	
	\item \label{edge:chng3} Topology changes: Upon indications of edge insertions/deletions, for each such insertion/deletion $\{v,u\}$, the edge $\{v,u\}$ is enqueued into $Q_{v,i}$ along with a corresponding insertion/deletion mark. In case of deletions the attached number is $0$.
	
	\item \label{edge:comm3} Communication: If $Q_{v,i}$ is not empty, the node $v$ dequeues an item from $Q_{v,i}$ and broadcasts it. Node $v$ also broadcasts a Boolean indication $\IsEmpty$ of whether $Q_{v,i}$ was empty in the beginning of round $i$, and an indication $\AreNeighborsEmpty$ of whether $v$ received $\IsEmpty$ from \emph{all} of its neighbors in the end of round $i-1$. In actuality, we do not send $\IsEmpty=\true$: not receiving $\IsEmpty=\false$ by other nodes is interpreted as receiving $\IsEmpty=\true$. The same rule holds for $\AreNeighborsEmpty=\true$.
	
	\item \label{edge:updt3} Updating the data structure: Upon dequeueing an item from $Q_{v,i}$ and broadcasting it, or receiving an item from a neighbor $u$, if it has an insertion mark then it is a path $p$. Denote by $p'$ the path
	\[
	p'=\begin{cases}
	p,&\quad p=\{\{v,w\}\}\text{ for some node }w,\\
	v-p,&\quad \text{otherwise}.
	\end{cases}
	\] 
	In this case for every edge $e\in p'$, the subpath $p''\subseteq p'$ leading to $e$ along $p'$ is added into $P_e$ in $S_{v,i}$. Furthermore, if $p'$ is an edge or a $2$-path, $p'$ is enqueued on $Q_{v,i}$.
	
	Alternatively, suppose the item broadcasted by $v$ or received from a neighbor $u$ has a deletion mark. In this case the item is an edge $e$ with an attached number $\ell$. Then, for every edge $e'$ in $S_{v,i}$, if $P_{e'}$ contains a path which includes $e$, it is removed from $P_{e'}$. Furthermore, if $\ell\leq 1$, then $e$ is enqueued on $Q_{v,i}$ with an attached number $\ell+1$.
	
	In both cases, if either $Q_{v,i}$ is not empty, or an item with $\IsEmpty=\false$ or $\AreNeighborsEmpty=\false$ is received, then $C_{v,i}$ is set to $\false$. If this was not the case for rounds $i$ and $i-1$, $C_{v,i}$ is set to $\true$.
\end{enumerate}
~\\\textbf{Correctness:}
Suppose $DS_{v,i}$ is queried with $\{w,x\}$. Then $DS_{v,i}$ responds $\inconsistent$ if $C_{v,i}$ is set to $\false$, and otherwise it responds $\true$ if and only if $\{w,x\}$ has a nonempty $P_{\{w,x\}}$. We need to show that if $C_{v,i} = \true$ then $DS_{v,i}$ responds $\true$ if the edge is in $R^{v,3}_{i-1}$ and $\false$ if the edge is not in $E^{v,3}_{i-1}$. It will be more convenient to actually prove the claim for edges in the $2$-hop neighborhood of $v$ with respect to round $i$, and for the remaining edges in the $3$-hop neighborhood of $v$ with respect to round $i-1$. We will show this can be done without loss of generality. Recall that the robust $2$-hop neighborhood, $R_i^{v,2}$, is defined such that $v-u-w\subseteq R^{v,2}_{i}$ if $t_{\{u,w\}}\geq t_{\{v,u\}}$, and $E^{v,1}_i\subseteq R_i^{v,2}$. Formally, we will actually prove that $R^{v,2}_{i}\cup \left(R^{v,3}_{i-1}\setminus R^{v,2}_{i-1}\right)\subseteq \tilde{S}_{v,i}\subseteq E_{i}^{v,2}\cup \left(E^{v,3}_{i-1}\setminus E^{v,2}_{i-1}\right)$. Note that we want to prove that $R^{v,3}_{i-1}\subseteq\tilde{S}_{v,i}\subseteq E^{v,3}_{i-1}$, but this follows if we have $DS_{v,i}$ respond on queries of edges in the $2$-hop neighborhood according to $\tilde{S}_{v,i-1}$ and the remaining edges in the $3$-hop neighborhood according to $\tilde{S}_{v,i}$.

First, assuming that $\{w,x\}\notin E_{i}^{v,2}\cup \left(E^{v,3}_{i-1}\setminus E^{v,2}_{i-1}\right)$ we show that $\{w,x\}\notin\tilde{S}_{v,i}$. By definition, we need to show that there is no path $p$ (including $\{w,x\}$) in $S_{v,i}$ (and consequently $P_{\{w,x\}}$) in round $i$. Note also that if a path $p$ is not in $S_{v,i}$ then all superpaths $p'\supset p$ are not in $S_{v,i}$ as well. We divide into cases depending on the length of $p$, and show that no paths of length $1$, $2$ or $3$ exist in $P_{\{w,x\}}$, thereby proving $\{w,x\}\notin \tilde{S}_{v,i}$.

\begin{enumerate}
	\item If $|p|=1$ then $v\in\{w,x\}$ and our assumption is that $\{w,x\}\notin E_{i}^{v,1}$. Since $v\in\{w,x\}$, the \emph{last deletion of} $\{w,x\}$ was enqueued on $Q_{v,i}$. Clearly,  $\{w,x\}\notin \tilde{S}_{v,i}$ since $Q_{v,i}$ must be empty for $C_{v,i}=\true$, implying that by round $i$ the path $p$ was removed from $S_{v,i}$.
	\item If $|p|=2$ then $p=v-w-x$ and our assumption is that $\{w,x\}\notin E_{i}^{v,2}$. If $\{v,w\}\notin E_{i}^{v,1}$ then we are in the previous case. Suppose alternatively that $\{v,w\}\in E_{i}^{v,1}$. Note also that since $C_{v,i}=\true$ it holds that both $Q_{v,i}$ and $Q_{w,i}$ are empty. Now, if $v$ was informed by $w$ about the deletion of $\{w,x\}$ we are done. If $v$ was not informed by $w$ about the deletion of $\{w,x\}$ it means that on the round $j<i$ (since $Q_{w,i}$ is empty) when $w$ dequeued the \emph{last copy of the deletion of} $\{w,x\}$, $\{v,w\}\notin E_{j}^{v,1}$, and on some round $j<j'\leq i$, $\{v,w\}\in E_{j'}^{v,1}$. Since $Q_{w,j'}$ and onwards is promised to no longer include items with $\{w,x\}$, $p$ is never included again into $S_{v,i}$. 
	\item If $|p|=3$ then $p=v-u-w-x$ and our assumption is that $\{w,x\}\notin E^{v,3}_{i-1}$. If $\{v,u\}$ or $\{u,w\}$ is not in $E_{i}^{v,2}$ then we are in the previous case. Suppose alternatively that $\{v,u\},\{u,w\}\in E_{i}^{v,2}$. Note also that since $C_{v,i}=\true$ it holds that $Q_{v,i}$, $Q_{u,i}$ and $Q_{w,i-1}$ are empty. Now, if $v$ was informed by $u$ about the deletion of $\{w,x\}$ we are done. If $v$ was not informed by $u$ about the deletion of $\{w,x\}$ there are several cases to consider. 
	\begin{enumerate}
		\item Suppose that $w$ informed $u$ about the \emph{last copy of the deletion of} $\{w,x\}$ on some round $j''< i-1$ (since $Q_{w,i-1}$ is empty). If $v$ was not informed by $u$ about the deletion of $\{w,x\}$ it means that on the round $j''<j<i$ (since $Q_{u,i}$ is empty) when $u$ dequeued (last copy of) the deletion of $\{w,x\}$, $\{v,u\}\notin E_{j}^{v,1}$, and on some round $j<j'\leq i$, $\{v,u\}\in E_{j'}^{v,1}$. Since $Q_{u,j'}$ and onwards is promised to no longer include items with $\{w,x\}$, $p$ is never included again into $S_{v,i}$.
		
		\item Suppose that $w$ did not inform $u$ about the deletion of $\{w,x\}$. Since $C_{v,i}=\true$ if $\IsEmpty=\false$ was not received for two consecutive rounds, we know that $\{u,w\}\notin E^{v,3}_{i-1}$. Thus by an argument similar to the one for paths of length $2$, the path $u-w-x$ is never included into $S_{u,i-1}$. This implies $u$ could not send $\{w,x\}$ by round $i$, implying $p$ is never included again into $S_{v,i}$
	\end{enumerate}
\end{enumerate}

Second, assuming that $\{w,x\}\in R^{v,2}_{i}\cup \left(R^{v,3}_{i-1}\setminus R^{v,2}_{i-1}\right)$ we show that $\{w,x\}\in\tilde{S}_{v,i}$. By definition, we need to show that there is at least one path $p$ (including $\{w,x\}$) in $S_{v,i}$ (and consequently $P_{\{w,x\}}$) in round $i$. We divide into cases depending on the length of $p$ which exists in $E^{v,2}_{i}\cup \left(E^{v,3}_{i-1}\setminus E^{v,2}_{i-1}\right)$ and show that for each case at least one path exists in $P_{\{w,x\}}$, thereby proving $\{w,x\}\in\tilde{S}_{v,i}$.

\begin{enumerate}
	\item If $|p|=1$ then $v\in\{w,x\}$ and our assumption is that $\{w,x\}\in E_{i}^{v,1}\subseteq R^{v,2}_{i}$. Since $v\in\{w,x\}$, the \emph{last insertion of} $\{w,x\}$ was enqueued on $Q_{v,i}$. Clearly in this case $\{w,x\}\in \tilde{S}_{v,i}$ since $Q_{v,i}$ must be empty for $C_{v,i}=\true$, implying that by round $i$ the path $p$ was included in $S_{v,i}$.
	\item If $|p|=2$ then $\{w,x\}\in R^{v,2}_{i}$. In this case $p=v-w-x\subseteq E_{i}^{v,2}$ and $t_{\{v,w\}}\leq t_{\{w,x\}}$. We also have that $t_{\{w,x\}}< i$ since $C_{v,i}=\true$ (which implies that both $Q_{v,i}$ and $Q_{w,i}$ are empty). This implies $\{w,x\}$ was enqueued by $w$ in round $t_{\{w,x\}}$ and that $\{v,w\}\in E_{t_{\{w,x\}}}^{v,1}$. Consequently, since $Q_{w,i}$ is empty, $v$ will be informed by $w$ about the insertion of $\{w,x\}$ and so $p\in P_{\{w,x\}}$ in $S_{v,i}$.
	\item If $|p|=3$ then $\{w,x\}\in R^{v,3}_{i-1}$. In this case $p=v-u-w-x\subseteq E_{i}^{v,2}\cup \left(E^{v,3}_{i-1}\setminus E^{v,2}_{i-1}\right)$ and $t_{\{v,u\}},t_{\{u,w\}}\leq t_{\{w,x\}}$. Since $Q_{w,i-1}$ is empty we know that $t_{\{w,x\}}< i-1$ and $u$ was informed by $w$ about $\{w,x\}$ by round $i-1$. Since $Q_{u,i}$ is also empty we know that $v$ was informed by $u$ about $u-w-x$ by round $i$ implying $p\in P_{\{w,x\}}$ in $S_{v,i}$.
\end{enumerate}

~\\\textbf{Round complexity:}
A topology change in the pair $e=\{w,u\}$ in round $i$ causes an enqueue of an item to $Q_{u,i}$ and to $Q_{w,i}$. For nodes $v_1,\ldots,v_\ell$ which are neighbors of $u$ this can further cause an enqueue in round $i_u>i$ of an item to $Q_{v_1},\ldots,Q_{v_\ell}$, respectively. The situation is similar for $w$ with respect to round $i_w>i$. Since $e$ caused an enqueue on at most $3$ rounds and nodes dequeue a single element every round we then have, for every round $j$, that the number of rounds in which there exists at least one node $v$ with an inconsistent $DS_v$ until round $j$ is bounded by $3$ times the number of topology changes which occurred until round $j$. 
This gives the claimed $O(1)$ amortized round complexity.
\end{proof}

\section{Lower bound for $k$-cycle listing whenever $k\geq 6$}
\label{section:cycleLB}

\begin{theorem-repeat}{theorem:cycleLB}
	\TheoremCycleLB
\end{theorem-repeat}

\begin{proof}
	Let $t$ and $D$ be integers to be specified later. Denote $\gamma=\lceil k/2\rceil-1$. We will consider the counter example on $\gamma t+tD=n$ nodes where we have nodes $\{u_i^j\}_{(i,j)\in[t]\times[\gamma]}$ and $\{v_i^j\}_{(i,j)\in [t]\times[D]}$. The adversary will proceed in two phases. In phase I, for $\ell=1,\ldots,t$, it will arbitrarily connect $u^1_\ell$ to exactly $2D/3$ nodes out of $\{v_\ell^j\}_{j\in [D]}$. It will then connect all of $\{v_\ell^j\}_{j\in [D]}$ to $u_\ell^2$, and further connect the path $u_\ell^2-\ldots - u_\ell^{\gamma}$. In phase II (illustrated in Figure \ref{fig:cycleLB}), it performs the following steps for $\ell=1,\ldots,t$:
	\begin{enumerate}
		\item For every $m=1,\ldots,\ell-1$:
		\begin{enumerate}
			\item Connect $u_\ell^1$ to $u_m^1$ and $u_\ell^{\gamma}$ to $u_m^{\gamma}$
			\item Wait for the algorithm to stabilize.
			\item Disconnect $u_\ell^1$ from $u_m^1$ and $u_\ell^{\gamma}$ from $u_m^{\gamma}$.
		\end{enumerate}
		\item If $\lfloor k/2\rfloor<\lceil k/2 \rceil$: Disconnect $u_\ell^{\lfloor k/2\rfloor-2}$ from $u_\ell^{\lceil k/2\rceil-2}$ and $u_\ell^{\lceil k/2\rceil-2}$ from $u_\ell^{\gamma}$. Then connect $u_\ell^{\lfloor k/2\rfloor-2}$ to $u_\ell^{\gamma}$.
	\end{enumerate}

	\begin{figure}[htbp]
		\begin{center}
			\includegraphics[scale=0.5, trim = {5cm 9cm 5cm 5cm}, clip]{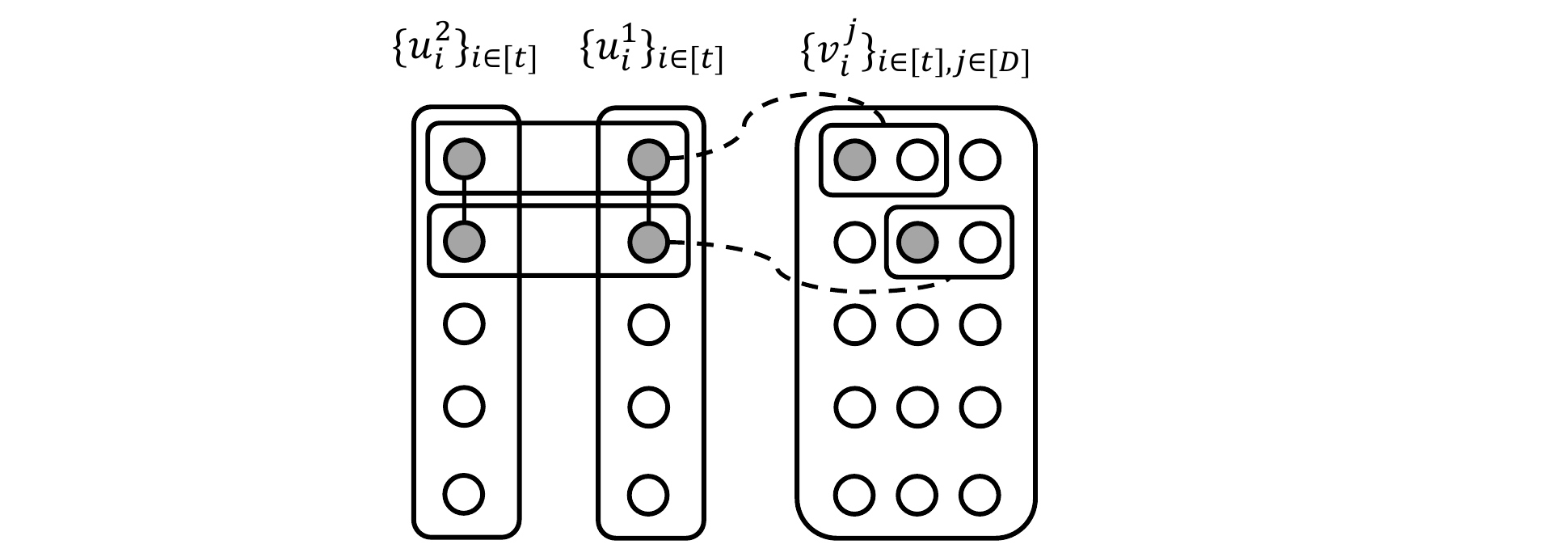}
		\end{center}
		\caption{\label{fig:cycleLB}Illustration of Phase II of the adversary in the proof of Theorem 5 for $k=6$ ($\gamma=2$) and $n=25$ ($t=5$, $D=3$), in iteration $\ell=2$ and subiteration $m=1$. The node $u^2_1$  ($u^2_2$) is connected to all 3 nodes in its 'row' $\{v^j_1\}_{j\in [D]}$ ($\{v^j_2\}_{j\in [D]}$), while $u^1_1$  ($u^1_2$) is connected to 2 of those. The nodes $u^1_1$, $u^1_2$, $u^2_1$, $u^2_2$ create many 6-cycles, with all possible nodes in $v^j_i$ that are connected to them. An example is given in gray.}
	\end{figure}
	
	Choosing $t=D+\gamma=\sqrt{n}$, we claim that no algorithm can handle this scenario with $o\left(\frac{\sqrt{n}}{\log n}\right)$ amortized round complexity. Indeed, consider the time we connect $u_\ell^1$ and $u_\ell^\gamma$ to $u_m^1$ and $u_m^\gamma$, respectively, for some $m<\ell$, and wait for the algorithm to stabilize. Before proceeding we make the following observation: from the outputs of the entire component $C_\ell:=\{u_\ell^j\}_{j\in[\gamma]}\cup \{v_\ell^j\}_{j\in [D]}$ we can deduce at least $D/6$ nodes from $\{v_m^j\}_{j\in [D]}$ which are connected to $u_m^1$, or from the outputs of the entire component $C_m:=\{u_m^j\}_{j\in[\gamma]}\cup \{v_m^j\}_{j\in [D]}$ we can deduce at least $D/6$ nodes from $\{v_\ell^j\}_{j\in [D]}$ which are connected to $u_\ell^1$. This is because, by construction, there are at least $D/3$ indices $J\subseteq [D]$ such that for every $j\in J$, both $v_\ell^j$ and $v_m^j$ are connected to $u_\ell^1$ and $u_m^1$, respectively. Then, by correctness, for each such index $j\in J$ from the outputs of either $C_\ell$ or $C_m$ we must deduce it is indeed the case that both $v_\ell^j$ and $v_m^j$ are connected to $u_\ell^1$ and $u_m^1$, respectively. This is because for any $j\in [D]$ such that the $k$-cycle
	\[
	v_\ell^j-u_\ell^1-u_m^1-v_m^{j}-u_m^2-\ldots -u_m^{\lfloor k/2\rfloor-2}-u_m^\gamma-u_\ell^\gamma-u_\ell^{\lceil k/2\rceil-2}-\ldots-u_\ell^2-v_\ell^j
	\]
	exists in the graph, by correctness, we will also deduce it exists from the outputs of either $C_\ell$ or $C_m$. Therefore, by a Pigeonhole argument, from the outputs of either $C_\ell$ or $C_m$ we must deduce at least $D/6$ such indices. Since each component had originally $\binom{D}{2D/3}$ configurations, and now from the output of at least one component (of $C_\ell$ or $C_m$) we can reduce the number of possible configurations of the other component to at most $\binom{5D/6}{D/2}$, it implies at least one component received at least $\log\binom{D}{2D/3}-\log\binom{5D/6}{D/2}=\Omega(D)$ bits by the algorithm. Crucially, it is not clear yet on which edges this happened.
	
	Now, we consider iteration $\ell$ where component $C_\ell$ connects to components $C_1,\ldots,C_{\ell-1}$. We will divide these components into two groups. Let $M=\{m_1,\ldots,m_h\}\subseteq [\ell-1]$ be indices such that for every $m\in M$, component $C_\ell$ received $\Omega(D)$ bits, and let the remaining indices $[\ell-1]\setminus M$ be such that for every $m\in [\ell-1]\setminus M$, component $C_m$ received $\Omega(D)$ bits due to connecting to component $C_\ell$. In addition, denote by $I_{m,\ell}$ the number of bits sent on the edges $\{u_m^1,u_\ell^1\}$ and $\{u_m^\gamma,u_\ell^\gamma\}$. 
	
	Clearly, for $m\in[\ell-1]\setminus M$, we have $I_{m,\ell}\geq \Omega(D)$, simply because whenever each such $C_m$ is connected to $C_\ell$, it must receive the $\Omega(D)$ bits on the edges $\{u_m^1,u_\ell^1\},\{u_m^\gamma,u_\ell^\gamma\}$, as it has no chance to receive these bits before these edges exist. This is because when $C_m$ had connections to other components they were never connected to $C_\ell$.
	
	Similarly, we have that $\sum_{m=1}^{m_1}I_{m,\ell}\geq \Omega(D)$. Next, by assumption, from the output of $C_\ell$ we can reduce the number of possible configurations of $C_{m_1}$ and $C_{m_2}$ from $\binom{D}{2D/3}^2$ to at most $\binom{5D/6}{D/2}^2$, implying $C_\ell$ received at least $2\left(\log\binom{D}{2D/3}-\log\binom{5D/6}{D/2}\right)=2\cdot \Omega(D)$ bits by the algorithm. Moreover, these bits could have been received only on the edges $\{\{u_m^1,u_\ell^1\},\{u_m^\gamma,u_\ell^\gamma\}\}_{m=1}^{m_2}$. Therefore, we have $\sum_{m=1}^{m_2}I_{m,\ell}\geq 2\cdot \Omega(D)$. Repeating this argument yields that $\sum_{m=1}^{m_h}I_{m,\ell}\geq h\cdot \Omega(D)$. This, together with $\sum_{m\in [\ell-1]\setminus M}I_{m,\ell}\geq (\ell-1-h)\cdot \Omega(D)$, implies that $2(I_{1,\ell}+\ldots +I_{\ell-1,\ell})\geq (\ell-1)\cdot \Omega(D)$.
	
	Consequently, we have shown that on iteration $\ell$ the total communication is bounded from below by $\Omega(\ell D)$. Since there are $t$ such iterations, it implies the total communication is $\Omega(t^2D)$. As there were only $O(t^2+tD)$ topological changes, and communication happend on only two edges at a time with $O(\log n)$ capacity, the amortized round complexity is
	\[
	\Omega\left(\frac{t^2D}{(t^2+tD)\log n}\right)=\Omega\left(\frac{\sqrt{n}}{\log n}\right).
	\]
\end{proof}

\begin{remark}
It is possible to modify the lower bound proof for $6$-cycle listing to obtain a similar lower bound for $3$-path listing. This is done by unifying $u_\ell^1$ and $u_\ell^\gamma$ into a single node and then connecting to it exactly $2D/3$ nodes from $\{v_\ell^j\}_{j\in[D]}$ as before. This means that already for some $4$-vertex subgraphs we cannot obtain ultra fast graph listing in the highly dynamic setting.
\end{remark}


\bibliographystyle{plainurl}
\bibliography{bibl}

\appendix

\section{Warm-up: The robust $2$-hop neighborhood}\label{sec:robust1hop}
\label{section:robust} 

Let $E^{v,r}_i$ denote the subset of $E_i$ of all edges contained in the $r$-hop neighborhood of $v$. We define a notion of robustness of an edge in $E^{v,2}_i$ with respect to $v$, as follows. We associate every edge $e$ with a value $t_e$ that we call its \emph{insertion time}, which is the latest round number in which $e$ was inserted (initially $t_e=-1$). Since insertion times can grow arbitrarily large, we stress that they are not part of any algorithm and are defined only for the sake of analysis. We say that an edge $e=\{u,w\}$ in $G_i$ is \emph{$(v,i)$-robust} if $v$ is one of its endpoints, or  $t_e \geq t_{\{v,u\}}$ and $\{v,u\} \in G_i$, or $t_e \geq t_{\{v,w\}}$ and $\{v,w\} \in G_i$. Now, instead of requiring that each node $v$ learns all of $E^{v,2}_i$, we require that $v$ learns all edges that are $(v,i)$-robust. To this end, we denote by $\Rvi$ the set of $(v,i)$-robust edges. 

Formally, the \emph{robust $2$-hop neighborhood listing} problem requires the data structure $DS_v$ at each node $v$ to respond to a query of the form $\{u,w\}$ with an answer $\true$ if the edge is $(v,i)$-robust, $\false$ if it is not $(v,i)$-robust, or $\inconsistent$, if $DS_v$ is in an inconsistent state. 
Recall that the node $v$ is not allowed to use any communication for deciding on its response. We claim the following.
\begin{theorem}
\label{theorem:N1}
There is a deterministic distributed dynamic data structure for the robust $2$-hop neighborhood listing, which handles edge insertions and deletions in $O(1)$ amortized rounds.
\end{theorem}

\begin{proof}
	
The data structure $DS_{v,i}$ at node $v$ at the end of round $i$ consists of the following: A set $S_{v,i}$ of items, where each item is an edge $e=\{a,b\}$ along with a timestamp $t'_e$, a queue $Q_{v,i}$ of items, each of which is an edge along with an insertion/deletion mark, and a flag $C_{v,i}$. We make a distinction between $t_e$ which is the true timestamp when $e$ was added and $t'_e$ which is the imaginary timestamp maintained inside $S_{v,i}$. Note that for every $e$ adjacent to $v$, the node $v$ knows the value $t_e$, while for non adjacent edges we only know $t'_e$.
We will make sure that an edge $\{a,b\}$ appears in at most a single item in $S_{v,i}$. Our goal is to maintain that $\Rvi= S_{v,i} $ at the end of round $i$, or that $C_{v,i}=\false$ (the consistency flag). 

Initially, for all nodes $v$, we have that $S_{v,0}$ and $Q_{v,0}$ are empty, and $C_{v,0}=\true$, indicating that $DS_{v,0}$ is consistent.
The algorithm for a node $v$ in round $i\geq 1$ is as follows:
\begin{enumerate}
	\item \label{edge:init} Initialization: Set $S_{v,i} = S_{v,i-1}$ and $Q_{v,i} = Q_{v,i-1}$. 
	
	\item \label{edge:chng} Topology changes: Upon indications of edge deletions, for each such deletion $\{v,u\}$, the edge $\{v,u\}$ is removed from $S_{v,i}$. 
	Then, for each such deletion $\{v,u\}$, all edges $\{\{u,z\}\in S_{v,i} \mid \text{either }\{v,z\} \not\in S_{v,i} \text{ or }t'_{\{u,z\}} < t_{\{v,z\}} \}$ are removed from $S_{v,i}$. Afterwards, upon an indication of an edge insertion $\{v,u\}$, the edge $\{v,u\}$ is added to the set $S_{v,i}$.
	In both cases (insertion and deletion), the pair $\{v,u\}$ is enqueued into $Q_{v,i}$ along with a corresponding insertion/deletion mark.

	\item \label{edge:comm} Communication: 
	If $Q_{v,i}$ is not empty, the node $v$ dequeues an item $e$ from $Q_{v,i}$ and sends it to all of its neighbors $u$ such that $t_e \geq t_{\{ v,u\}}$, along with a Boolean indication $\IsEmpty$ of whether $Q_{v,i}$ is now empty or not. In actuality, we do not send $\IsEmpty=\true$: not receiving $\IsEmpty=\false$ by other nodes is interpreted as receiving $\IsEmpty=\true$.

	\item \label{edge:updt} Updating the data structure: Upon receiving an item $e=\{u,w\}$ from a neighbor $u$, node $v$ sets updates $S_{v,i}$ according to the insertion/deletion mark. Furthermore, for the case of insertion, if $e\notin S_{v,i-1}$ we set $t'_e=t_{\{u,v\}}$. Otherwise we update $t'_e=\max \{t'_e, t_{\{u,v\}}\}$.
	If $Q_{v,i}$ is not empty, or an item with $\IsEmpty=\false$ is received, 
	then $C_{v,i}$ is set to $\false$. Otherwise, $C_{v,i}$ is set to $\true$.
\end{enumerate}

~\\\textbf{Correctness:} Suppose $DS_{v,i}$ is queried with $\{u,w\}$. Then $DS_{v,i}$ responds $\inconsistent$ if $C_{v,i}$ is set to $\false$, and otherwise it responds $\true$ if and only if $\{u,w\} \in S_{v,i}$. We need to show that if $C_{v,i} = \true$ then $DS_{v,i}$ responds $\true$ if the edge is $(v,i)$-robust and $\false$ otherwise. We will first prove correctness for the \emph{ideal algorithm}, where we set $t'_e = t_e$ (the true edge timestamp). Then, we show that our algorithm behaves exactly the same as the ideal algorithm, even with the modified timestamps.

To show the correctness of the ideal algorithm, we want to show that if $C_{v,i}=\true$, it holds that $S_{v,i}=\Rvi$. Let $e=\{u,z\}\in \Rvi$, and we show that if $C_{v,i}=\true$ then $e$ is also in $S_{v,i}$. If $v\in e$ the proof is direct, so we focus on the case where $v \notin e$. Because $e\in \Rvi$ there must exist an edge $e'=\{u,v\} \in \Rvi$ such that $t_{e} \geq t_{e'}$. That is, the edge $e'$ was added at the same round or before $e$, and remained there until the $i$-th round. As $C_{v,i}=\true$ implies that the queue of $u$ is empty, this means that $v$ must have received the edge $e'$ by iteration $i$ and added it to $S_{v,i}$. As $e'$ was unchanged throughout this time, this implies that the edge $e$ was not deleted from $S_{v,i}$, up to and including round $i$.

Next, let $e=\{u,z\}\notin \Rvi$, and we show that if $C_{v,i}=\true$ then $e$ is also  not in $S_{v,i}$. Here again the case where $v\in e$ is trivial. Furthermore, if at the time when $e$ was added both $\{u,v\},\{z,v\}$ did not exist, it will never be sent to $v$ due to the condition in step~\ref{edge:comm}. Thus, we are only concerned with the case where $e$ was added to $S_{v,j}$ for $j<i$, but is not in $\Rvi$. This implies that there exists no $e'\in \Rvi$ such that $e\cap e' \neq \emptyset$ and $t_{e} \geq t_{e'}$. If such an edge $e'$ would exist, this implies that $e$ is $(v,i)$-robust via the definition of a robust neighborhood. By step~\ref{edge:chng} of our algorithm, this implies that $e$ must be deleted by the $i$-th iteration. This completes the proof for the ideal algorithm.

 We need to show that the set $S_{v,i}$ maintained by our algorithm is the same as it would be if the value of $t'_e$ was set to $t_e$ (the ideal algorithm). For insertions of edges the timestamp does not play a role in the algorithm, thus we must only concern ourselves with deletions. When an edge $\{u,v\}$ is deleted, we remove all edges $\{\{u,z\}\in S_{v,i} \mid \text{either }\{v,z\} \not\in S_{v,i} \text{ or }t'_{\{u,z\}} < t_{\{v,z\}} \}$. For the condition $\{v,z\} \not\in S_{v,i}$ the timestamp is irrelevant, thus we must only consider the case where $\{v,z\} \in S_{v,i}$ and $t'_{\{u,z\}} < t_{\{v,z\}}$. Let us consider what is implied when this condition holds in the ideal algorithm. This implies that $\{v,z\}$ was added \emph{after} $\{u,z\}$. This in turn means that our algorithm will also remove this edge as $t'_{\{u,z\}}=t_{\{u,v\}} < t_{\{v,z\}}$. On the other hand if the edge is kept because $t'_{\{u,z\}} \geq t_{\{v,z\}}$, this means that $\{v,z\}$ was added \emph{before} $\{u,z\}$ and thus $t'_{\{u,z\}}\geq t_{\{v,z\}}$. This is because the queues of all neighbors are empty, and thus the value of $t'_{\{u,z\}}$ will be set to at least $t_{\{v,z\}}$.
In both cases our algorithm acts exactly the same as the ideal algorithm and thus it holds that $S_{v,i}=\Rvi$ when $C_{v,i} = \true$.

~\\\textbf{Round complexity:}
A topology change in the pair $e=\{w,u\}$ in round $i$ causes an enqueue of an item to $Q_{u,i}$ and to $Q_{w,i}$. We then have, for every round $j$, that the number of rounds in which there exists at least one node $v$ with an inconsistent $DS_v$ until round $j$ is bounded by the number of topology changes which occurred until round $j$. 
This gives the claimed $O(1)$ amortized round complexity.
\end{proof}


\section{$2$-hop neighborhood listing in $O\left(n/\log n\right)$ amortized complexity}\label{appendix:1hop}

\begin{lemma}\label{lem:up1hop}
	There is a deterministic distributed dynamic data structure for $2$-hop neighborhood listing which handles edge insertion and deletions in $O\left(\frac{n}{\log n}\right)$ amortized rounds.
\end{lemma}

\begin{proof}	
	Consider an algorithm where a node $v$ maintains a separate update queue $Q_u$ for each of its neighbors $u\in N_v$, with the goal that each neighbor will know all changes done to $N_v$. Upon edge deletion $\{u,v\}$, the endpoints $v$ and $u$ both enqueue the deletion event $\{u,v\}$ on the update queues for all their neighbors, incurring $O(1)$ amortized complexity. For edge insertion $\{u,v\}$,  the endpoints $v$ and $u$  also enqueue a single item on each update queue for all their neighbors. Furthermore, upon such an insertion, each of the endpoints, $v$ and $u$, takes a snapshot of its neighborhood, which is an $O(n)$ bit string, and enqueues it on the update queue of the other endpoint (which is equivalent to enqueuing $O(n/\log n)$ items). Since every topology change causes an enqueue of at most $O(n/\log n)$ items for each queue, and nodes dequeue a single element from each queue every round, the amortized round complexity is $O(n/\log n)$.
	
\end{proof}

\begin{remark}
	By Combining Lemma \ref{lem:up1hop} with Theorem \ref{thm:lbmembershiplisting} it follows that membership listing for all $2$-diameter graphs can be handled in $O\left(\frac{n}{\log n}\right)$ amortized rounds. 
\end{remark}

\end{document}